\documentclass[conference]{IEEEtran}
\usepackage[cmex10]{amsmath}
\usepackage{amssymb}
\usepackage{cite}
\usepackage{graphicx}

  \graphicspath{{./figures/}}

\usepackage{algorithmic}

\usepackage{setspace}

\usepackage{array}

\usepackage{fixltx2e}

\usepackage{color}

\usepackage{times,color}

\usepackage{epstopdf}

\newcommand{\mc}{\mathcal}

\newcommand{\ms}{\mathsf}

\newtheorem{theorem}{Theorem}

\newtheorem{definition}{Definition}

\begin{document}
\title{Many-Broadcast Channels: Definition and Capacity in the Degraded Case}

\author{Tsung-Yi Chen\IEEEmembership{Member,~IEEE,}, Xu Chen~\IEEEmembership{Student Member,~IEEE,}and Dongning Guo~\IEEEmembership{Senior~Member,~IEEE}\\
Department of Electrical Engineering and Computer Science\\
Northwestern University, Evanston, IL 60208, USA
}

\maketitle

\begin{abstract}
\boldmath Classical multiuser information theory studies the fundamental limits of models with a fixed (often small) number of users as the coding blocklength goes to infinity.  Motivated by emerging systems with a massive number of users, this paper studies the new {\em many-user paradigm}, where the number of users is allowed to grow with the blocklength.  The focus of this paper is the degraded many-broadcast channel model, whose number of users may grow as fast as linearly with the blocklength. A notion of capacity in terms of message length is defined and an example of Gaussian degraded many-broadcast channel is studied. In addition, a numerical example for the Gaussian degraded many-broadcast channel with fixed transmit power constraint is solved, where every user achieves strictly positive message length asymptotically. 
\end{abstract}

\section{Introduction}
\label{sec:Intro}
Multiuser information theory studies the fundamental limits of communication systems with multiple sources, transmitters and/or receivers. The capacity region is characterized by studying the asymptotic regime with coding blocklength $n$ growing to infinity for a fixed number of users $k$. The theory lays the foundation of designing multiuser systems such as cellular networks and wireless ad hoc networks. Several prior works that study large systems also consider the case where the number of users $k$ is taken to infinity \textit{after} the blocklength is taken to infinity. 

In general, the theory that assumes a fixed number of users does not apply to systems where the number of users is comparable or even larger than the blocklength, such as in some sensor networks or machine-to-machine (M2M) communication systems with many thousands of devices
in a given cell. A key reason is that for many functions $f(k,n)$, letting $k\to\infty$ after $n\to\infty$ may yield a different
result than letting $n$ and $k=k_n$ (as a function of $n$)
simultaneously tend to infinity,\footnote{Take the function $f(n,k)=\log(1+k/n)$ as an example. Taking the limits separately gives $0$ or $\infty$ while taking the limit simultaneously with $k_n = n$ yields $\lim\limits_{n\to\infty} f(n,n) = \log 2$.} i.e., 
\begin{align} \label{eq:fkn}
\lim_{k\to\infty}\lim_{n\to\infty} f(k,n) \ne \lim_{n\to\infty} f(k_n,n)\,.
\end{align}
This new paradigm in multiuser information theory models where $k_n$ can grow arbitrarily large with $n$ is referred to as the many-user regime.

One motivating example is the design of ultra-scalable M2M communication systems where the number of users $k$ is comparable or even larger than the blocklength $n$, and the message transmitted to each user could be very short. The many-user regime therefore becomes a better performance indicator in the context of M2M communication where $k_n = O(n)$ and the number of bits to be transmitted for each user may be sub-linear in $n$. We are interested in the fundamental limits in this regime. Yet the rate for each user vanishes as $k_n$ grows, indicating that the traditional notion of capacity in bits per channel use becomes ill-suited for the task.

M2M communication represents an example where we must review the transition from multiuser to many-user systems carefully. Similar effects have been observed before in the multiuser information theory literature as Cover and Thomas, for example, noted in~\cite[p.~546]{cover1991elements} for the Gaussian multiaccess channel with per-user power constraint, ``when the total number of senders is very large, so that there is a lot of interference, we can still send a total amount of information that is arbitrary large even though the rate per individual sender goes to 0.'' Similar effects appear in Gaussian broadcast channels as we increase the number of users. Therefore it is crucial to identify a suitable notion of capacity in order to understand the fundamental limits in the many-user paradigm.

A many-user channel model that parallels the Gaussian multiaccess channels, referred to as the Gaussian many-access channel, was studied in \cite{Chen_gaussian_2013}. The symmetric capacity was defined in terms of the message lengths. An achievability scheme using maximum-likelihood (ML) decoding for the scaling of $k_n = o(n)$ was shown in \cite{Chen_gaussian_2013} and the scaling of $k_n = O(n)$ was shown in \cite{chen2014gaussian}. This paper studies the fundamental limit of the message length for degraded broadcast channels (DBC) in the many-user regime, referred to as the degraded many-broadcast channels.

The rest of the paper is organized as follows: Section~\ref{sec:Notataion} introduces the notation and definitions used throughout the paper. Section~\ref{sec:MatinResults} gives the main results and provides an example Gaussian degraded many-broadcast channel. A sketch of technical proof is presented in the Appendix. Finally, Section~\ref{sec:Conclusion} concludes the paper.

\section{Notations and Definitions}
\label{sec:Notataion}
Uppercase letters represent random variables and the associated realizations are denoted by lowercase letters. The only exception is $M$, which denotes the number of codewords in a codebook. This paper focuses on the class of memoryless channels.

Because the limit of the channel coding rate is ill-suited for our purpose, we shall study the fundamental limit of the message length instead. We begin with, consider a single-user discrete memoryless channel (DMC) described by $P_{Y|X}$ and then generalizes to the many-user models. An $\left(n, M, \epsilon\right)$ code for the channel $P_{Y|X}$ consists of $M$ codewords of blocklength $n$ with error probability no greater than $\epsilon$. Denote the maximal codebook size with error probability no greater than $\epsilon$ and blocklength $n$ as
\begin{align}
M^*(n, \epsilon) = \max\{ M: \exists \text{ an $(n, M, \epsilon)$-code}\}.
\end{align}
Letting $C(n) = nC$, where $C$ is the capacity of the DMC, the classical achievability and strong converse for channel coding can be restated as follows: 1) For all $\delta > 0$, there exists a vanishing sequence $\epsilon_n \to 0$ as $n\to\infty$ such that
\begin{align}
\label{eqn:msg_ach}
\lim_{n\rightarrow\infty}\frac{\log_2 M^*(n, \epsilon_n)}{C(n)} \ge 1-\delta.
\end{align}
2) For all vanishing sequences $\epsilon_n$:
\begin{align}
\label{eqn:msg_converse}
\limsup_{n\rightarrow\infty}\frac{\log_2 M^*(n, \epsilon_n)}{C(n)} \le 1.
\end{align}
Therefore we can define the fundamental limit of the message length, referred to as the message length capacity $\mc{C}$, as a collection of sequences $C(n)$ such that for each $C(n) \in \mc{C}$ the two conditions regarding \eqref{eqn:msg_ach} and \eqref{eqn:msg_converse} hold.
One way of characterizing $\mc{C}$ is to use the asymptotic notation: $\mc{C} = \{nC + o(n)\}$, which is a collection of functions rather than a fixed number.

We can also define the message length capacity alternatively based on the following notion of achievable message length:
\begin{definition}[Achievable message length]
\label{def:message_ach}
A sequence of message length $\lfloor \log_2 M_n \rfloor$ indexed by a positive integer $n$ is said to be asymptotically achievable if there exist a sequence of $(n, M_n, \epsilon_n)$ codes such that  $\epsilon_n\to 0$ as $n\to\infty$. 
\end{definition}

Throughout the paper the notion of achievability is always in the asymptotic sense. An equivalent definition of the message length capacity is the follow:
\begin{definition}[Message length capacity]
\label{def:message_cap}
$\mc{C}$ is a collection of all sequences $C(n)$ such that for any $\delta > 0$, the message length sequence $\lfloor(1-\delta)C(n)\rfloor$ is asymptotically achievable and for any $\delta > 0$, $\lceil(1+\delta)C(n)\rceil$ is not asymptotically achievable.
\end{definition}

A $k$-user memoryless broadcast channel with input $X$ and $k$ outputs $Y_1, \dots, Y_k$ without feedback is described by the conditional probability $P_{Y_1\dots Y_k|X}$. The channel is a degraded broadcast channel if there exists a Markov chain $X-Y_1- \dots -Y_k$ that yields consistent marginals $P_{Y_j|X}, j= 1, \dots, k$ with $P_{Y_1\dots Y_k|X}$~\cite{Bergmans_random_1973}. If the input and output alphabets are finite, then it is a discrete memoryless degraded broadcast channel (DM-DBC). 
The capacity region remains open for general broadcast channels, while the capacity region is known for the class of DBC due to seminal works by Cover~\cite{Cover_broadcast_1972}, Bergman~\cite{Bergmans_random_1973}, and Gallager~\cite{Gallager_degraded_1974}. Since the two-user DBC can be generalized to $k$-user DBC, we only state the result for two-user DBC. 

Let $R_j$ be the rate of the $j$th user and $U$ be an auxiliary random variable. The capacity region of a two-user DM-DBC is known to be the set of rates $(R_1, R_2)$ satisfying
\begin{align}
\label{eqn:CapDBC1}
R_{2} &\leq I(U; Y_2)
\\
\label{eqn:CapDBC2}
R_1 &\leq I(X; Y_1|U)
\end{align}
for some $P_{U, X}$ and $U$ has cardinality no greater than $\min\{|\mc{X}|, |\mc{Y}|\} + 1$.

A degraded many-broadcast channel has a number of receivers $k_n$ growing as a function of the blocklength $n$. The channel consists of an input space $\mc{X}$, a sequence of output spaces $\mc{Y}^{k_n}$ and a sequence of memoryless channels $P_{Y_1 Y_2 \dots Y_{k_n}|X}$ indexed by $n$. Conditioned on channel input $x^n \in\mc{X}^n$, the channel outputs of the memoryless channel has product probability measure  $\prod_{j = 1}^n P_{Y_1 Y_2 \dots Y_{k_n}|X = x_j}$ on the product space $\mc{Y}^{k_n\times n}$. 
\begin{definition}
\label{def:mnbc_code}
 An $\left(n, \{M_{j}\}_{j = 1}^{k_n}, \epsilon\right)$ many-broadcast code for a many-broadcast channel $P_{Y_1 Y_2\dots Y_{k_n}|X}$ consists of
\begin{enumerate}
\item An encoder $f:\mc{W}_1\times\dots\times\mc{W}_{k_n}\mapsto \mc{X}^n$ with $|\mc{W}_j| = M_j$.
\item $k_n$ decoders $g_j: \mc{Y}^n\mapsto \mc{W}_j$, $j = 1, \dots, k_n$, whose error probability satisfies 
\begin{align*}
\max\limits_{j\in\{1, \dots, k_n\}}P[W_j \ne g_j(Y_j^n)] \leq \epsilon,\end{align*}
where $W_1, \dots, W_k$ are independent uniform random variables on their respective alphabets. 
\end{enumerate}
\end{definition}
\begin{figure*}
\setcounter{equation}{7}
\begin{align}
\label{eqn:E0full}
E_{0}(\rho_{j, i}, P_{Y_j|U_i})
= \sum_{u_{i+1}\in\mc{U}_{i+1}}P_{U_{i+1}}(u_{i+1})\left[ \sum_{y_j \in\mc{Y}}
\left( \sum_{u_i \in\mc{U}_i}P_{U_i|U_{i+1}}(u_i|u_{i+1})P_{Y_j|U_{i}}(y_j|u_i)^{\frac{1}{1+\rho_{j,i}}}\right)^{1+\rho_{j,i}}\right]
\end{align}
\vskip-0.5em
\hrulefill
\vskip-0em
\setcounter{equation}{6}
\end{figure*}

Since we are considering a sequence of channels $P_{Y_1 Y_2 \dots Y_{k_n}|X}$ that can, in principle, be defined arbitrarily as the blocklength $n$ increases, it is necessary to restrict our attention to a sequence of ``regular'' channels. We will focus on the class of degraded many-broadcast channels where each marginal channel $P_{Y_j|X}$, $j = 1, 2, \dots, k_n$, for each $n$ follows the same class of distributions. 
One example of a regular sequence of channels is the memoryless Gaussian degraded many-broadcast channel. With a given blocklength $n$, the received signal of the $k_n$ users in some symbol interval are given by:
\begin{align}
\label{eqn:gaussian_mnbc}
Y_j = X + \sigma_{n, j}Z_j, \quad j = 1, \dots, k_n\,,
\end{align} 
where $Z_j\sim\mc{N}(0, 1)$ and $\sigma_{n, j}$ denotes the standard deviation of the noise. The noise levels form a triangular array, and without loss of generality, we assume $\sigma_{n, j} \le \sigma_{n, j+ 1}$ for all $n$ and $j = 1, \dots, k_n-1$. For the Gaussian degraded many-broadcast channel with power constraint $\gamma$, we have an additional constraint in Definition \ref{def:mnbc_code} that every codeword $x^n$ must satisfy $\sum_{i = 1}^{n}x_i^2 \leq n\gamma$.

Following Definition \ref{def:message_ach}, we define the achievable message lengths for many-broadcast channels:
\begin{definition}[Achievable array for many-broadcast]
A triangular array $\lfloor \log_2 M_{n,j} \rfloor$ indexed by integers $n$ and $j = 1, 2, \dots, k_n$ is said to be an asymptotically achievable message length array if there exist a sequence of $\left(n, \{M_{n,j}\}_{j = 1}^{k_n}, \epsilon_n\right)$ many-broadcast code such that $\epsilon_n\to 0$ as $n\to\infty$. 
\end{definition}

To simplify the presentation, we use $M_j$ to denote the element of the triangular array $M_{n,j}$ whenever it is clear that $M_j$ is a function of $n$. Matching the definition for single-user channel, we denote the elements in a triangular array $C_{n, j}$ as $C_j(n)$. 
\begin{definition}[Message length capacity for many-broadcast]
\label{def:mnbc_message_cap}
The message length capacity $\mc{C}$ for a many-broadcast channel is a collection of triangular arrays $C_{j}(n)$ such that for every $\delta > 0$, $\lfloor(1-\delta)C_{j}(n)\rfloor$ is asymptotically achievable and $\lceil(1+\delta)C_{j}(n)\rceil$ is not asymptotically achievable.
\end{definition}

\section{Main Results}
\label{sec:MatinResults}

The scaling of $k_n$ with $n$ distinguishes degraded many-broadcast from the conventional DBC setting. With a fixed number of users in a DBC, time-sharing can achieve a significant amount of the full capacity region, especially when the channel conditions of different users are similar.  In Gaussian many-broadcast channel with $k_n = n$ users, however, applying time-sharing scheme fails to achieve reliable communication: Since each user will only have a single channel use even as $n$ grows arbitrarily large, the error probability cannot vanish with the blocklength in general.

This paper studies the possible growth rate of $k_n$ and the corresponding coding scheme for achieving the message length capacity. A degenerate case in a many-broadcast channel is when the channels to all receivers are statistically identical, so that all users can decode all messages, if any at all. The broadcast is then equivalent to a single-user communication where a user first decodes the commonly decodable message and then finds her own segment in it. In such degenerate case it is trivial to see that $k_n = O(n)$ is achievable and each user can transmit a constant number of bits asymptotically.


The general capacity results for degraded many-broadcast closely resemble the capacity region of the conventional DBC, but the achievability using typical set decoding does not immediately extend to the many-user regime \cite{Chen_gaussian_2013}. An ML decoding analysis for DBC \cite{Gallager_degraded_1974} is used to show the achievability. A sketch of proof is presented in the Appendix.
\subsection{Discrete Memoryless Channels}
Consider a discrete memoryless degraded many-broadcast channel $\{P_{Y_j|X}\}_{j = 1}^{k_n}$. Let $E_0(\rho_{j,i}, P_{Y_j|U_i})$ be given as \eqref{eqn:E0full}, shown at the top of the page. We need the following technical conditions related to the error exponent analysis of \cite{Gallager_degraded_1974}: Assume that the second derivative of $E_0(\rho_{j,i}, P_{Y_j|U_i})$ with respect to $\rho_{j,i}$ is continuously differentiable and bounded, i.e., there is a constant $\kappa < \infty$ such that
\setcounter{equation}{8}
\begin{align}
\label{eqn:diff_condition}
\left|\frac{\partial^2 E_0(\rho_{j,i}, P_{Y_j|U_i})}{\partial \rho_{j,i}^2} \right| < \kappa\,.
\end{align}
 We have the following theorem:
\begin{theorem}
\label{thm:MainResult}
Consider a discrete memoryless degraded many-broadcast channel that satisfies conditions regarding \eqref{eqn:diff_condition}. Let $U_1 = X$, $U_{k_n + 1} = 0$. 
For some admissible $X,U_2, \dots, U_{k_n}$ that form a Markov chain: 
\begin{align}
U_{k_n}-U_{k_n-1}-\dots-U_2-X-Y_1-Y_2-\dots-Y_{k_n},
\end{align} 
let the triangular array $C_j(n)$ be given as:
\begin{align}
\label{eqn:MainResult1}
C_j(n) &= nI(U_{j}; Y_j|U_{j+1}),\quad j = 1, 2,\dots, k_n.
\end{align}
Then $C_j(n)$ is an admissible message length capacity if for all $\delta > 0$ and $j$ (either a fixed finite index or $j= \beta k_n$, $\beta \in (0,1]$), 
\begin{align}
\label{eqn:sufficient_cond}
n\delta I(U_{j}; Y_{j}|U_{j+1})^2 - \log k_n
\end{align}
is unbounded as $n\to\infty$.
\end{theorem}	

The first term in \eqref{eqn:sufficient_cond} is related to the single-user channel error exponent of ML decoding, where as the second term is due to the union bound over $k_n$ users' error event. A sketch of the proof is given in the Appendix. 

Using Gaussian random codebook for the Gaussian degraded many-broadcast channel, \eqref{eqn:diff_condition} holds (see for example \cite{chen2014gaussian}) and hence the proof is similar to the DMC case. Note that we still need to check \eqref{eqn:sufficient_cond} for a given triangular array of the noise levels and a power allocation among the users.

\subsection{Gaussian Degraded Many-Broadcast with $k_n = O(n)$}
We study an example of Gaussian degraded many-broadcast channel in this subsection. The channel model for a Gaussian degraded many-broadcast with a total power constraint $\gamma$ and $k_n$ users is given as by \eqref{eqn:gaussian_mnbc}. Let $\alpha = (\alpha_1, \dots, \alpha_{k_n})$ be a non-negative vector such that $\sum_{j = 1}^{k_n}\alpha_j = 1$. We choose  $P_{U_j} \sim \mc{N}(0, \alpha_j\gamma)$ independent of each other and let $X = \sum_{j = 1}^{n} U_j$. For the many-user regime we consider the case with uniform power allocation. The same derivation in the following applies to the case when $\alpha_j = O(1/k_n)$. We assume that the triangular array of the noise variances satisfies $\sigma_{n, j+1} \ge \sigma_{n, j} \ge \epsilon > 0 $ and 
$\lim\limits_{n\to\infty}\sigma_{n, k_n}= \sigma < \infty$.

For $k_n = o(\sqrt{n})$, e.g., $k_n = n^{1/3}$, we can verify that \eqref{eqn:sufficient_cond} holds by Taylor's expansion:
\begin{align}
\sqrt{n}\delta I(U_j;Y_j|U_{j+1})
&=\frac{\sqrt{n}\delta}{2}\log\left(1+\frac{\frac{\gamma}{n^{1/3}}}{\sigma_{n,j}^2+\frac{(j-1)\gamma}{n^{1/3}}}\right)
\\
&=\frac{n^{1/6}\delta\gamma/2}{\sigma_{n, j}^2+ \frac{(j-1)\gamma}{n^{1/3}}} - O(n^{-1/6})\,.
\end{align}
In order for superposition coding and successive decoding for $k_n = O(n)$ to work, we need an additional step of grouping, referred to as grouped superposition coding. The proof of the following theorem presents the grouping procedure. 
\begin{theorem}
\label{thm:ManyGaussian}
Assume that the degraded Gaussian many-broadcast channel satisfies 
\begin{align}
\label{eqn:dense_assumption}
|\sigma_{n, j} - \sigma_{n, j+o(\sqrt{n})}| = o(1).  
\end{align}
The following triangular array is in the message length capacity of a Gaussian many-broadcast channel with $k_n = O(n)$, power constraint $\gamma$ and uniform power allocation:
\begin{align}
\label{eqn:CapAWGN}
C_j(n) = \frac{n}{2}\log_2\left(1+\frac{\gamma/k_n}{\sigma_{n,j}^2+(j-1)\gamma/k_n}\right)\,,
\end{align}
where $j$ is either a fixed finite index or $j = \beta k_n$, $\beta \in (0,1]$.
\end{theorem}
\begin{proof}
The converse holds as in the case for DMC. To simplify the notation, assume without loss of generality that $k_n = n$. Let $a \in (0,1/2)$ and $\bar{a} = 1-a$. Sequentially group $n^{\bar{a}}$ users to share the same codebook and consider each group as a super user indexed by $i = 1, 2, \dots, n^{a}$. The $i$th super user generates an i.i.d. Gaussian codebook according to $\mc{N}(0, \gamma/n^{a})$ for the channel $P_{Y_{in^{\bar{a}}}|X}$. In other words, the super user generates a codebook with message length that is suitable for the worst channel of the group. Therefore every user in the same group can decode the same codeword and can find segments of the message that belongs to oneself.  

The grouped superposition coding scheme yields an equivalent Gaussian many-broadcast channel with $n^{a}$ users: $\tilde{Y}_{i} = X + \tilde{\sigma}_{n, i}\tilde{Z}_{i}$ where $\tilde{\sigma}_{n, i} = \sigma_{n, i n^{\bar{a}}}$, $\tilde{Z}_{i} = Z_{in^{\bar{a}}}$. Hence the message length capacity of the super users includes the following triangular array since \eqref{eqn:sufficient_cond} holds:
\begin{align}
nI(\tilde{U}_{i};\tilde{Y}_{i}|\tilde{U}_{i+1}) 
\label{eqn:group_msglength}
&= \frac{n}{2}\log\left(1+\frac{\frac{\gamma}{n^{a}}}{\sigma_{n,i n^{\bar{a}}}^2 + \frac{(i-1)\gamma}{n^{a}}}\right).
\end{align}

Equally distributing the message length to each user in a group, we obtain a quantized triangular array where every $n^{\bar{a}}$ users have the same message length. In other words, for all $\delta > 0$ the triangular array $(1-\delta)\tilde{C}_j(n)$ is achievable where $\tilde{C}_j(n)$ is given as
\begin{align}
\label{eqn:group_array}
\tilde{C}_j(n)=\frac{n^a}{2}\log\left(1+\frac{\frac{\gamma}{n^a}}{\sigma_{n,\lceil j/n^{\bar{a}}\rceil n^{\bar{a}}}^2 + \frac{(\lceil j/n^{\bar{a}}\rceil-1)\gamma}{n^a}}\right).
\end{align}

The last step is to show that $(1-\delta)\tilde{C}_j(n)$ achievable implies that $(1-\delta)C_j(n)$ is also achievable asymptotically. By Taylor's expansion we have:
\begin{align}
\tilde{C}_j(n) = \frac{\gamma/2}{\sigma_{n,\lceil j/n^{\bar{a}}\rceil n^{\bar{a}}}^2 + \frac{(\lceil j/n^{\bar{a}}\rceil-1)\gamma}{n^a}} + O(n^{-a}).
\end{align}
Hence the difference between $\tilde{C}_{i n^{\bar{a}}}(n)$ and $C_{i n^{\bar{a}}}(n)$ is in the order of $O(n^{-a})$. Therefore $\forall \delta > 0$, $\exists\delta' \in (0, \delta)$ and a $n_0$ large enough such that $(1-\delta')\tilde{C}_{in^{\bar{a}}}(n) \ge (1-\delta)C_{in^{\bar{a}}}(n)$ for all $n \ge n_0$. 
Finally, by \eqref{eqn:dense_assumption} we have for $j = in^{\bar{a}}$:
\begin{align}
|C_{j - n^{\bar{a}} + 1}(n) - C_{j}(n)| = o(1)\,.
\end{align}
In other words, the quantization error goes to zero. Hence $\forall\delta > 0$, $(1-\delta)C_{j}(n)$ is achievable asymptotically, finishing the proof. 
\end{proof}

As a numerical example, let $k_n = cn$ for some $c > 0$ and let $\sigma_{n,j}^2 = 2^{j/k_n}$. Some asymptotic results follow immediately from Theorem~\ref{thm:ManyGaussian} assuming $\alpha_j = 1/k_n$, i.e., a degraded Gaussian many-broadcast channel with uniform power allocation. Taking the limit in $n$ for any fixed $j< \infty$ yields
\begin{align}
\lim_{n\to\infty}C_j(n) = \frac{\gamma\log_2(e)}{2c}\, \text{ bits }.
\end{align}
For $j = \beta k_n, \beta \in (0, 1]$ we have
\begin{align}
\lim_{n\to\infty}C_{\beta k_n}(n) = \frac{\gamma\log_2(e)}{2c(2^{\beta} + \beta \gamma)}\, \text{ bits }.
\end{align}

To provide a numerical example, consider the case when $k_n = n/4$, $\alpha_j = 1/k_n$. Figure~\ref{fig:GaussianMnBC} shows the plot of $C_{\beta k_n}(n)$ for a fixed $n = 1000$ and the limit as $n\to\infty$. The curve shows that asymptotically, all users can receive more than $2$ bits and $20\%$ of the users can receive more than $10$ bits. 

\begin{figure}
\centering
\includegraphics[width=0.45\textwidth]{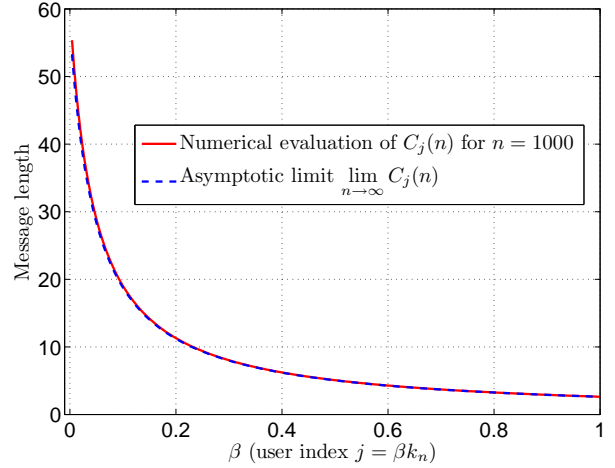}
\caption{Gaussian many-broadcast example for $k_n = n/4$, $\alpha_j = 1/k_n$ and $\sigma_{n,j}^2 = \exp\{j/k_n\}$. 
 }
\label{fig:GaussianMnBC}
\vskip-1em
\end{figure}

\section{Concluding Remarks}
\label{sec:Conclusion}
This paper proposed the degraded many-broadcast channel models and a new notion of capacity for such models. The capacity of discrete memoryless degraded many-broadcast and Gaussian degraded many-broadcast are presented. An example of Gaussian degraded many-broadcast channel where each user achieves strictly positive message length is provided.

\section*{Appendix: Sketch of Proof}
\label{sec:Proof}
\subsection{Converse}
The converse follows from the same converse argument of DBC converse with fixed number of users: By Fano's inequality we have for all $j$:
\begin{align}
\log M_j  &= H(W_j)
\\
&= I(W_j; Y_j^n) + H(W_j|Y_j^n)
\\
&\leq I(W_j;Y_j^n) + \epsilon^{(j)}_n\log M_j + H_2(\epsilon_n^{(j)})\,,
\end{align}
where $H_2(x) = -x\log x - (1-x)\log(1-x)$ is the binary entropy.
Following the standard technique for choosing the auxiliary random variables along with a uniform time-sharing random variable, we have
\begin{align}
\log M_j \leq n I\left(U_j; Y_j \big\vert U_{j+1}\right) + \epsilon^{(j)}_n\log M_j + H_2(\epsilon_n^{(j)})\,.
\end{align}
If $M_j$ is scaling with $n$ such that for $n$ large enough,
\begin{align}
\log M_j \geq (1+\delta)nI(U_j;Y_j|U_{j+1})\,,
\end{align}  
then we have
\begin{align}
\epsilon^{(j)}_n &\geq 1 - \frac{H_2(\epsilon_n^{(j)})}{\log M_j} - \frac{nI(U_j;Y_j|U_{j+1})}{\log M_j}
\\
&\geq 1 - \frac{H_2(\epsilon_n^{(j)})}{(1+\delta)nI(U_j;Y_j|U_{j+1})} - \frac{1}{1+\delta}
\\
&=1- \frac{H_2(\epsilon_n^{(j)})+nI(U_j;Y_j|U_{j+1})}{(1+\delta)nI(U_j;Y_j|U_{j+1})}\,.
\end{align}
Suppose that $\epsilon_n^{(j)} \to 0$ as $n\to\infty$, then taking limit on the both side in $n$ yields $\epsilon_n^{(j)} \geq \delta/(1+\delta) > 0$, a contradiction.  This finishes the converse for $k_n = O(n)$. 

\subsection{Achievability for Degraded Many-Broadcast}
For the stated condition for degraded many-broadcast channels, the following proof follows closely to the proof in \cite{Gallager_degraded_1974}:

\subsubsection{Codebook generation} Let $\ms{U} = (U_2, \dots, U_{k_n})$ and $X = U_1$. Fix a distribution $P_{\ms{U},X}$ satisfying
\begin{equation}
P_{\ms{U},X} = P_{U_{k_n}} \prod_{j = 1}^{k_n-1} P_{U_{j}|U_{j+1}}\,.
\end{equation}
Randomly generate a i.i.d. length-$n$ sequence $u_{k_n}(w_{k_n})$ according to $P_{U_{k_n}}$ for each message in $w_{k_n} \in \mc{W}_{k_n}$. For each of the generated sequence $u_j(w_{j})$ randomly generate i.i.d. layers of satellite sequences $u_{j-1}\left(w_{j-1}, u_k(w_{j})\right)$ according to $P_{U_{j-1}|U_{j}}$ for $j = k_n, \dots, 3$. Finally generate the i.i.d. input sequences $x^n(w_1, w_2)$ randomly according to $P_{U_1|U_{2}} = P_{X|U_{2}}$.
\subsubsection{Encoding and decoding} To send the message set $(w_1, w_2, \dots, w_{k_n})$, first encode $w_{k_n}$ to $u_{k_n}^n(w_{k_n})$. Then for $j = k_n , \dots, 3$, sequentially generate $u_{j-1}^n\left(w_{j-1}, u_j^n(w_{j})\right)$. Finally the transmitter sends $x^n\left(w_1, u_{2}^n(w_2)\right)$. The $j$th receiver performs ML decoding and successively cancels layers of satellite codewords starting from the $k_n$ th codeword to the $j$th codeword, $j = 1, \dots, k_n$.
\subsubsection{Error analysis} 
Let the error probability of the $j$th user be $\epsilon_n^{(j)}$. Using union bounds on the error events of the successive decoding rule, we have: 
\begin{align}
\label{eqn:UnionBound}
\epsilon_n^{(j)} \leq \sum_{i = j}^{k_n} \epsilon_n^{(j,i)}\,,
\end{align}
where $\epsilon_n^{(j,i)}= P_{Y_j}[\hat{w}_i \ne w_i]$, i.e., the decoding error probability of the $i$th user's codeword over the $j$th channel. Using the error exponent analysis in \cite{Gallager_degraded_1974} and \cite{Gallager_information_1968}, we have for any $\rho_{j,i} \in [0, 1], j = 1, \dots, k_n$ and $i = j, j+1, \dots, k_n$:
\begin{align}
\epsilon_n^{(j,i)} \leq \exp\left\{- E_{j,i}^{(n)}(\rho_{j, i}, M_i)\right\}\,,
\end{align}
where $E_{j,i}^{(n)}(\rho_{j, i}, M_i)$ for the $i$ message over the $j$th channel is given as (recall that $U_1 = X$ and $U_{k_n+1} = 0$): 
\begin{align}
\label{eqn:exponent}
E_{j,i}^{(n)}(\rho_{j, i}, M_i) = nE_{0, j, i}(\rho_{j, i}, M_i)- \rho_{j, i} \log M_i\,,
\end{align}
where $E_{0}(\rho_{j, i}, P_{Y_j|U_j})$ is given in \eqref{eqn:E0full}.
Optimizing \eqref{eqn:exponent} over $\rho_{j, i} \in [0,1]$ we obtain $n$ times the error exponent in \cite{Gallager_degraded_1974}:
\begin{align}
E_{j,i}^{*(n)}(M_j) = \sup\limits_{\rho_{j,i} \in [0, 1]}E_{j,i}^{(n)}(\rho_{j, i}, M_i)\,.
\end{align}
Since the channel is degraded, it suffices to focus on the admissible scaling of $M_j$ such that $E_{j,j}^{*(n)}(M_j)$ tend to infinity as $n\to\infty$. Similar to the single user channel coding, we have the following properties
\begin{align}
E_0(\rho_{j, j}, P_{Y_j|U_j}) &\ge 0 ; \,\rho_{j, j}\ge 0\,.
\\
-\kappa\le\frac{\partial^2 E_{0}(\rho_{j, j}, P_{Y_j|U_j})}{\partial \rho_{j, j}^2} &\leq 0; \,\rho_{j, j}\in [0, 1]\,.
\\
\left.\frac{\partial E_{0}(\rho_{j, j}, P_{Y_j|U_j})}{\partial \rho_{j, j}}\right\vert_{\rho_{j, j} = 0} &= I(U_j;Y_j|U_{j+1})\,.
\end{align}
The lower bound of the second property is from the assumption \eqref{eqn:diff_condition}. Fix a $\delta > 0$, let $\log M_j \leq n(1-\delta)I(U_j;Y_j|U_{j+1})$ for all $n$ large enough. For any $\rho_{j,j} \in [0,1]$, we have:
\begin{align}
E_{j,j}^{*(n)}(M_j) \ge nE_{0}(\rho_{j,j}, P_{Y_j|U_j})- \rho_{j,j} \log M_j\,.
\end{align}
By Taylor's expansion of $E_{0}(\rho, P_{Y_j|U_j})$ around the origin, there exists a $\rho_{j}' \in (0, \rho_{j,j})$ such that 
\begin{align}
&nE_{0}(\rho_{j,j}, P_{Y_j|U_j})- \rho_{j,j} \log M_j \nonumber
\\
&= n\left(\rho_{j,j}\delta I(U_j;Y_j|U_{j+1}) + \frac{\rho_{j,j}^2}{2}\frac{\partial^2 E_0(\rho_j', P_{Y_j|U_j})}{\partial\rho^2}\right)
\\
&\ge
n\rho_{j,j}\delta I(U_j;Y_j|U_{j+1})\left(1 - \frac{\kappa\rho_{j,j}}{2\delta I(U_j;Y_j|U_{j+1})}\right) \,.
\end{align}
Choosing $\rho_{j,j} = \frac{2\delta(1-\delta) I(U_j;Y_j|U_{j+1})}{\kappa}$ we have
\begin{align}
\label{eqn:rhojj_condition}
E_{j,j}^{*(n)}(M_j) \ge n\delta' I(U_j;Y_j|U_{j+1})^2 \,,
\end{align}
where $\delta' = \delta^3(1-\delta)/\kappa$ is a positive constant.
Since there are at most $k_n$ terms in \eqref{eqn:UnionBound}, each $\epsilon_n^{(j)}$ vanishes to zero in the order of $O(k_n\exp\{-n\delta'I(U_j;Y_j|U_{j+1})^2\})$, which goes to zero by the assumption $n\delta'I(U_j;Y_j|U_{j+1})^2 - \log k_n \to\infty$ as $n\to\infty$. This finishes the proof of the achievability. \QED

\bibliographystyle{IEEEtran}
\bibliography{IEEEabrv,ManyBroadcast}


\end{document}